\documentclass[runningheads,a4paper]{llncs}
\usepackage{tikz}
\usepackage{pgfplots}
\pgfplotsset{compat=1.18}
\usepackage[T1]{fontenc}
\usepackage[utf8]{inputenc}
\usepackage{lmodern}
\usepackage{amsmath,amssymb,mathtools}
\usepackage[hidelinks]{hyperref}
\usepackage[nameinlink,capitalise]{cleveref}
\usepackage{microtype}

\newcommand{\OPT}{\operatorname{OPT}}
\newcommand{\score}{\operatorname{score}}
\newcommand{\FF}{\textnormal{\textsc{FarthestFirst}}}
\newcommand{\1}{\mathbf 1}

\newcounter{algorithm}

\title{Remote Matching: Exact-Cardinality Approximation and Tight UGC Hardness}
\titlerunning{Approximation and Hardness for Remote Matching}

\author{Arash Ahadi\inst{1}\thanks{Corresponding author.} \and
Morteza Alimi\inst{2} \and
Sharareh Alipour\inst{1} \and
Shayan Tayefeh\inst{3}}
\authorrunning{A. Ahadi et al.}

\institute{Tehran Institute for Advanced Studies (TeIAS), Khatam University, Iran\\
\email{aarash.ahadi.academic@gmail.com, sh.alipour@teias.institute}
\and
University of Augsburg, Augsburg, Germany\\
\email{morteza.alimi@uni-a.de}
\and
Sharif University of Technology, Iran\\
\email{shayantayefeh22@gmail.com}}

\begin{document}
\maketitle

\begin{abstract}

In the unrestricted max--min metric $T$-join problem, one seeks an even
terminal set $T$ maximizing the cost of a minimum $T$-join.  Iwata and
Ravi gave a factor-$3/2$ approximation for this problem.  We show that
this guarantee is tight under the Unique Games Conjecture: no
polynomial-time approximation with factor strictly smaller than $3/2$
exists under UGC.

We then consider the exact-cardinality variant, which prescribes an even
number \(k\) of terminals.  Writing \(p:=k/n\), we give a deterministic
polynomial-time \(\rho(p)\)-approximation for every feasible cardinality,
where

\[
\rho(p)=
\begin{cases}
7/2, & \makebox[1.5em][r]{$0$}<p\le2/7,\\
1/p, & 2/7\le p\le2/3,\\
1/[2(1-p)], & 2/3\le p\le7/8,\\
4, & 7/8\le p<1.
\end{cases}
\]
In particular, a factor-\(4\) approximation holds throughout the entire feasible cardinality range, the factor is at most \(7/2\) whenever \(0<k\le 6n/7\), and equals \(3/2\) at \(k=2n/3\). The algorithmic framework is based on optimal laminar cut packings,
their weighted tree representations, exact-cardinality rounding, and
tree dynamic programming.

\end{abstract}

\section{Introduction}
\label{sec:introduction}

Let $(V,d)$ be a finite rational metric space and let $n=|V|\ge2$.  For every
even set $T\subseteq V$, let $\mu(T)$ denote the cost of a minimum-weight
perfect matching on $T$; so $\mu(\varnothing)=0$.  A $T$-join is an edge
set whose odd-degree vertices are precisely the vertices of $T$.
Equivalently, $\mu(T)$ is the minimum cost of a $T$-join in the complete
graph on $V$ with edge costs $d$.  Minimum $T$-joins correct the degree parities of connected spanning
structures, as in Christofides' metric traveling-salesman algorithm and
the Chinese-postman problem~\cite{Christofides76,EJ73}.

The \emph{unrestricted max--min metric $T$-join problem} asks for an even
set $T\subseteq V$ maximizing $\mu(T)$, that is,

\[
  \Phi(V,d):=\max\{\mu(T):T\subseteq V,\ |T|\text{ even}\}.
\]

Iwata and Ravi~\cite{IR13} studied the weighted max--min $T$-join problem,
proved it NP-hard, and gave a factor-$3/2$ approximation based on a cut-packing upper
bound, uncrossing, and LP duality.  We use the convention that a factor
$\rho\ge1$ guarantees value at least $\OPT/\rho$; a value fraction is
its reciprocal.  Inapproximability statements concern deterministic
polynomial-time algorithms.  Their result left
open how closely this approximation ratio reflects the intrinsic
hardness of the problem.  We address this question already for positive
metric instances.  We give a reduction from Independent Set to
the problem.  Combining the resulting identity with known Independent
Set gaps shows that no polynomial-time algorithm can guarantee any fixed
value fraction strictly larger than
$(4+\sqrt2)/7\approx0.77$ unless $P=NP$, while under the Unique Games
Conjecture no fraction strictly larger than $2/3$ is possible.  Thus,
under UGC, the Iwata--Ravi guarantee is optimal.

We next turn to the prescribed-cardinality version.  For a given even integer
$k$ with $0\le k\le n$, the \emph{exact-$k$ remote matching problem} is

\[
  \OPT_k:=\max\{\mu(T):T\subseteq V,\ |T|=k\}.
\]

Remote matching belongs to metric diversity maximization: select $k$
points maximizing a distance-based objective~\cite{RRT91,CH01,IMMM14}.
For example, remote-clique uses the sum of pairwise distances,
remote-tree the minimum-spanning-tree weight, and remote-cycle the
minimum traveling-salesman tour cost.  These objectives, including
remote matching and remote pseudoforest, have also been studied through
composable coresets~\cite{IMMM14}.

Remote matching remained exceptional within this family for a long
time.  Chandra and Halld'orsson~\cite{CH01} gave an $O(\log k)$
approximation, while Mahabadi and Narayanan~\cite{MN23} resolved the
longstanding open question of a constant-factor approximation by giving
the first randomized $O(1)$-approximation for $k\le n/3$, with emphasis
on establishing a constant guarantee rather than optimizing its value.
Appendix~\ref{app:MN-factor4} gives an infinite family for which, under
a valid parity-repair rule, the expected approximation ratio of a single
execution tends to $4$.
 They also obtained constant-factor
composable coresets for remote matching and remote pseudoforest.
Alipour~\cite{Alipour26} further studied weighted max--min $T$-joins,
giving bounds and approximation results for prescribed cardinalities,
as well as additional structural and exact results for special weight
sets.

Moreover, the exact-cardinality problem is hard to approximate within any
factor strictly below $2$ unless $P=NP$~\cite{HIKT99,MN23}.

Our algorithms are deterministic, cover every feasible even
cardinality, and give guarantees that depend explicitly on the density
$p:=k/n$.  A key structural ingredient is a tight exact-cardinality
density profile for parity joins on weighted laminar trees.  We first
give a centered all-cut algorithm whose choice $m=n$ yields the
density-dependent guarantee $\min\{p,2(1-p)\}$, while the choice
$m=3k/2$ gives a uniform $2/7$ value guarantee throughout the
substantially larger range $k\le2n/3$, extending well beyond the range
$k\le n/3$ of the previous constant-factor algorithm.  A second,
allowed-cut algorithm gives a universal $1/4$ value guarantee for every
feasible even $k$.  Taking the best applicable algorithm yields
\noindent
\begin{minipage}[c]{0.56\textwidth}
\begin{equation}
\rho(p)=
\begin{cases}
7/2, & 0<p\le2/7,\\
1/p, & 2/7\le p\le2/3,\\
\dfrac{1}{2(1-p)}, & 2/3\le p\le7/8,\\
4, & 7/8\le p<1.
\end{cases}
\label{eq:intro-rho}
\end{equation}
\end{minipage}
\hfill
\begin{minipage}[c]{0.36\textwidth}
\centering
\begin{tikzpicture}
\begin{axis}[
    width=\linewidth,
    height=3.2cm,
    xmin=0, xmax=1,
    ymin=1, ymax=4.2,
    xlabel={$p=k/n$},
    ylabel={$\rho(p)$},
    axis lines=left,
    xtick={0,2/7,2/3,7/8,1},
    xticklabels={$0$,$\frac27$,$\frac23$,$\frac78$,$1$},
    ytick={1,2,3,4},
    yticklabels={$1$,$2$,$3$,$4$},
    tick label style={font=\scriptsize},
    label style={font=\scriptsize},
    line width=0.55pt,
    clip=false
]
\addplot[thick,domain=0:2/7] {7/2};
\addplot[thick,domain=2/7:2/3,samples=50] {1/x};
\addplot[thick,domain=2/3:7/8,samples=50] {1/(2*(1-x))};
\addplot[thick,domain=7/8:1] {4};
\addplot[only marks,mark=o,mark options={fill=white}] coordinates {(1,4) (0,3.5)};
\addplot[only marks,mark=*] coordinates {(0,1) (1,1)};
\end{axis}
\end{tikzpicture}
\end{minipage}
\par
The feasible endpoint cases $k=0,n$ are solved exactly.  Table~\ref{tab:intro-summary}
separates the constrained and unrestricted guarantees.

\begin{table}[t]
\centering
\caption{Approximation factors; the UGC bound is for the unrestricted problem.}
\label{tab:intro-summary}
\setlength{\tabcolsep}{14pt}
\begin{tabular}{lll}
\hline
 & Approximation & Inapproximability\\
\hline
Unrestricted
& $3/2$~\cite{IR13}
& $3/2-\varepsilon$ (this paper, UGC)\\[0.5ex]
Exact-$k$
& $\rho(p)$ (this paper)
& $2-\varepsilon$ (worst case over $k$)~\cite{HIKT99,MN23}\\
\hline
\end{tabular}
\end{table}

The first algorithm selects farthest-first representatives, packs all cuts
on their induced metric, and maximizes the packed odd-cut weight over
exact-$k$ subsets.  The second selects terminals or omitted vertices;
its allowed cuts also account for a fixed reference terminal set.
Concavity converts even-set averaging into exact-cardinality rounding.
The tools and algorithms appear in \cref{sec:framework,sec:approximation},
with structural proofs in \cref{sec:structural-proofs};
\cref{sec:hardness} proves the unrestricted hardness result.

\section{Hardness for the Unrestricted Problem}
\label{sec:hardness}

\begin{theorem}[Unrestricted inapproximability]
\label{thm:unrestricted-hardness}
For every fixed $\varepsilon>0$, the unrestricted max--min metric
$T$-join problem admits no polynomial-time approximation with factor
$7/(4+\sqrt2)-\varepsilon$ unless $P=NP$, and none with factor
$3/2-\varepsilon$ under the Unique Games Conjecture, even for metrics
with distances in $\{1,2,3\}$.
\end{theorem}

We prove the theorem in two steps.  First, a simple paired metric turns
the unrestricted optimum into an exact expression involving the
independence number.  We then transfer standard Independent Set gaps
through this identity.
For a graph $H$, let $\alpha(H)$ and $\alpha'(H)$ denote its independence
number and maximum-matching number, respectively.

Let $H=(W,E_H)$ have $2b$ vertices and a supplied perfect matching $M$.
Define a metric over $W$; for distinct $u,v\in W$, define
\begin{equation}
d_H(u,v)=
\begin{cases}
1, & uv\in M,\\
2, & uv\in E_H\setminus M,\\
3, & uv\notin E_H,
\end{cases}
\qquad d_H(v,v)=0.
\label{eq:paired-metric}
\end{equation}
This is a metric because a triangle contains at most one edge of $M$, so the sum
of any two edge lengths is at least $3$, while the third is at most $3$.

\begin{lemma}[Paired-metric identity]
\label{lem:paired-identity}
For the metric \eqref{eq:paired-metric},
\begin{equation}
  \Phi(W,d_H)=b+\left\lfloor\frac{\alpha(H)}2\right\rfloor.
\label{eq:paired-identity}
\end{equation}
\end{lemma}

\begin{proof}
Fix an even $T\subseteq W$.  Call a pair of $M$ full, singleton, or
empty according as it contains two, one, or zero vertices of $T$, and
let $f,s,e$ be the respective numbers of pairs.  Since $|T|=2f+s$,
$s$ is even.  Let $R$ contain the selected vertex from each singleton pair.

There is an optimum matching on $T$ that matches every full pair
internally.  Indeed, if a full pair $uv\in M$ is matched instead by
$ua$ and $vb$, those two edges cost at least $4$, while replacing them
by $uv$ and $ab$ costs at most $1+3=4$.  Repeating this exchange fixes
all full pairs without increasing the cost; an already fixed pair
cannot be disturbed by a later exchange.

It remains to match the $s$ vertices of $R$.  An edge inside $R$ has
length $2$ exactly when it lies in $H[R]$, and otherwise has length
$3$.  Thus each edge of $H[R]$ used in a perfect matching saves one unit,
and the maximum possible saving is $\alpha'(H[R])$: a maximum matching
of $H[R]$ leaves an independent set $I$ of
$s-2\alpha'(H[R])$ vertices, which can be paired arbitrarily.  Hence
\begin{equation}
 \mu(T)=f+\frac{3s}{2}-\alpha'(H[R])
       =b-e+\frac{|I|}{2}.
\label{eq:terminal-value}
\end{equation}
Since $e\ge0$ and $|I|\le\alpha(H)$, this implies
$\mu(T)\le b+\lfloor\alpha(H)/2\rfloor$.

Conversely, delete at most one vertex from a maximum independent set to
obtain an even independent set $I$ of size
$2\lfloor\alpha(H)/2\rfloor$.  Since $M\subseteq E_H$, the set $I$
contains at most one endpoint of each pair of $M$.  Let $T$ contain $I$
and both endpoints of every $M$-pair disjoint from $I$.  Then $T$ has no
empty pairs, its singleton representatives are exactly $I$, and
$\alpha'(H[I])=0$.
Equation~\eqref{eq:terminal-value} gives equality.
\end{proof}

\begin{lemma}[Independent Set gap transfer]
\label{lem:is-gap-transfer}
Fix $a\in(0,1/2]$.  Suppose that, for arbitrarily small fixed
$\delta>0$, it is NP-hard to distinguish $N$-vertex graphs with
\[
  \alpha(G)\ge(a-\delta)N
  \qquad\text{from}\qquad
  \alpha(G)\le\delta N.
\]
Then, for every fixed $\beta>1/(1+a)$, it is NP-hard to guarantee a
$\beta$ fraction of the unrestricted optimum, already for the metrics
\eqref{eq:paired-metric}.
\end{lemma}

\begin{proof}
Fix $\beta>1/(1+a)$ and choose a sufficiently small permitted
$\delta\in(0,a/2)$, as specified below.  Let $G$ be a gap instance.
Compute a maximal matching and let $U$ be its
unmatched vertices.  Since $U$ is independent, $|U|>\delta N$ certifies
a YES instance.  Otherwise set $G_0:=G-U$.  The maximal matching is a
perfect matching of $G_0$, and
$m:=|V(G_0)|\in[(1-\delta)N,N]$.

In the YES case,
$\alpha(G_0)\ge(a-2\delta)N\ge(a-2\delta)m$, whereas in the NO case,
$\alpha(G_0)\le\delta N\le\delta m/(1-\delta)$.  Form
$H:=G_0\mathbin{\dot\cup}G_0$ and supply it with the union of the two
perfect matchings.  Then $H$ has $2m$ vertices and
$\alpha(H)=2\alpha(G_0)$, so \cref{lem:paired-identity} gives exactly
\[
 \Phi(V(H),d_H)=m+\alpha(G_0).
\]
Hence
\[
 \Phi_{\rm YES}\ge(1+a-2\delta)m,
 \qquad
 \Phi_{\rm NO}\le\frac{m}{1-\delta}.
\]
Continuity at $\delta=0$ permits the choice of $\delta$ to satisfy
\[
 \beta(1+a-2\delta)>\frac1{1-\delta}.
\]
Thus $\beta\Phi_{\rm YES}>\Phi_{\rm NO}$.  A $\beta$-fraction algorithm would
therefore distinguish the two cases, since the value of its returned
terminal set is computable by minimum-weight perfect matching.
\end{proof}

\par
\begin{trivlist}
\item[\hskip\labelsep{\itshape Proof of \cref{thm:unrestricted-hardness}.}]
The completed Independent Set gap of Khot, Minzer, and Safra has
$a=1-1/\sqrt2$~\cite{KhotMinzerSafra2023,KhotMinzerSafra2025}.
Corollary~1.9 and Section~1.5 of~\cite{KhotMinzerSafra2025} give the gap
and explain the subsequent proof of its expansion hypothesis.  By \cref{lem:is-gap-transfer}, no polynomial-time
algorithm can guarantee a value fraction larger than
$1/(1+a)=(4+\sqrt2)/7$ unless $P=NP$.  Equivalently, no approximation
factor below $7/(4+\sqrt2)$ is possible.

Under UGC, the Khot--Regev gap~\cite{KhotRegev2008} has $a=1/2$.
The same lemma gives value threshold $2/3$, equivalently approximation
factor $3/2$.  
\end{trivlist}

\section{Algorithmic Tools and Preliminaries}
\label{sec:framework}

Throughout, $(V,d)$ is a finite rational metric space, $n:=|V|$, and
running times refer to the binary encoding model.  For nonempty
$X\subseteq V$, let $d_X:=d|_{X\times X}$ and $E_X:=\binom X2$; also put
$d(v,Y):=\min_{y\in Y}d(v,y)$ for nonempty $Y\subseteq V$.
For a designated subset $C\subseteq V$, we refer to its elements as
\emph{centers}.

A \emph{shore} of $X$ is a nonempty proper subset $S\subsetneq X$, with
cut $\delta_X(S):=\{uv\in E_X:|\{u,v\}\cap S|=1\}$.
The shores $S$ and $X\setminus S$ represent the same unoriented cut.
For each nonempty $X\subseteq V$, fix an arbitrary root $r\in X$ and
represent every cut by its unique shore not containing $r$.  Accordingly,
$\mathcal A(X):= \{S:\varnothing\ne S\subseteq X\setminus\{r\}\}$
is the family of representatives of all nontrivial cuts of $X$.  When
$X=V$, we omit the subscript from $E_X$, $d_X$, and $\delta_X$.  For
$T\subseteq X$, a cut represented by $S$ is \emph{$T$-odd} if
$|S\cap T|$ is odd.  For any proposition $\mathcal P$, the notation
$\1[\mathcal P]$ denotes its indicator.

The symmetric difference of a $P$-join and a $Q$-join is a
$(P\triangle Q)$-join.  Thus, for even $P,Q\subseteq V$, nonnegative
costs and the matching characterization give
\begin{equation}
 \mu(P\triangle Q)\le\mu(P)+\mu(Q).
 \label{eq:matching-triangle}
\end{equation}

\begin{lemma}[Projection to centers]
\label{lem:projection}
Let $D\subseteq V$ be even, let $C\subseteq V$, and let
$\pi:D\to C$.  If $E\subseteq C$ consists of the centers having an odd
number of preimages under $\pi$, then $E$ is even and
\[
  \mu(D\triangle E)
  \le \sum_{v\in D}d(v,\pi(v)).
\]
In particular, if $d(v,\pi(v))\le R$ for every $v\in D$, then
$\mu(D\triangle E)\le |D|R$.
\end{lemma}

\begin{proof}
Take the attachment edges $v\pi(v)$ modulo two.  The symmetric difference
of their endpoint sets is $D\triangle E$: the first endpoints contribute
$D$, the second contribute $E$, and a loop contributes nothing.  Thus
they form a $(D\triangle E)$-join of weight at most
$\sum_{v\in D}d(v,\pi(v))$.
\end{proof}

Fix a nonempty $X\subseteq V$ and a cut family
$\mathcal F\subseteq\mathcal A(X)$.  The associated cut-packing LP and
its covering dual are standard~\cite{EJ73,Schrijver03}

\noindent
\begin{minipage}[t]{0.57\textwidth}
\small
\refstepcounter{equation}\label{eq:cut-packing}
\[
\begin{aligned}
W(X,\mathcal F)=\max
&\sum_{S\in\mathcal F}y_S\\
\text{s.t.}\quad
&\sum_{\substack{S\in\mathcal F\\ e\in\delta_X(S)}}y_S
   \le d_X(e)
   \quad(e\in E_X),\\
&y_S\ge0
   \quad(S\in\mathcal F).
\end{aligned}
\tag{\theequation}
\]
\end{minipage}%
\hfill
\vrule 
\hfill
\begin{minipage}[t]{0.39\textwidth}
\small
\refstepcounter{equation}\label{eq:covering-dual}
\[
\begin{aligned}
\min
&\sum_{e\in E_X}d_X(e)x_e\\
\text{s.t.}\quad
&x(\delta_X(S))\ge1
   \quad(S\in\mathcal F),\\
&x_e\ge0
   \quad(e\in E_X).
\end{aligned}
\tag{\theequation}
\]
\end{minipage}
where
$x(\delta_X(S)):=\sum_{e\in\delta_X(S)}x_e$.
A vector $y$ satisfying \eqref{eq:cut-packing} is a \emph{feasible
packing}.  For such $y$ and even $T\subseteq X$, define the
\emph{packed parity score}
$\score_y(T)=\sum_{S\in\mathcal F}y_S \1[|S\cap T|\equiv1\pmod2].$

\begin{lemma}[Packing certificate~\cite{EJ73,Schrijver03}]
\label{lem:certificate}
Let $\mathcal F\subseteq\mathcal A(X)$ and let $y$ be feasible for
\eqref{eq:cut-packing}.  For every even $T\subseteq X$,
$\score_y(T)\le\mu(T)$.  If $\mathcal F$ contains the representative of
every $T$-odd cut of $X$, then
\[
  \mu(T)\le W(X,\mathcal F).
\]
In particular, $W(X,\mathcal A(X))\ge\mu(T)$ for every even
$T\subseteq X$.
\end{lemma}

\begin{proof}
Every perfect matching on $T$ crosses each $T$-odd cut.  Assign each
scored cut to a matching edge crossing it; summing edge packing constraints
gives $\score_y(T)\le\mu(T)$.  If $\mathcal F$ contains all $T$-odd cuts,
every feasible covering vector lies in the $T$-join dominant, whose minimum
cost is $\mu(T)$~\cite{EJ73,Schrijver03}.  Strong duality gives the second
inequality.
\end{proof}

Two shores $S,T\subsetneq X$ \emph{cross} if all four sets
$S\cap T$, $S\setminus T$, $T\setminus S$, and
$X\setminus(S\cup T)$ are nonempty.  A family of root-oriented shores is \emph{laminar} if its members
are pairwise disjoint or nested; equivalently, no two members cross.  A cut family $\mathcal F$ is \emph{uncrossable} if,
whenever two representing shores $S,T$ cross, either the cuts represented
by both $S\cap T,S\cup T$ or those represented by both
$S\setminus T,T\setminus S$ belong to $\mathcal F$.

For nonnegative capacities $c\in\mathbb Q_+^{E_X}$, write
$c(\delta_X(S)):=\sum_{e\in\delta_X(S)}c_e$.  For a packing vector $y$, write
$\operatorname{supp}(y):=\{S\in\mathcal F:y_S>0\}$.

The following standard uncrossing principle~\cite{HLST88,GLS88,IR13}
gives a polynomial-time computable laminar optimum.
\begin{lemma}[Computable laminar optimum]
\label{lem:laminarization}
Let $\mathcal F\subseteq\mathcal A(X)$.  If $\mathcal F$ is uncrossable
and a minimum-capacity member of $\mathcal F$ can be found in polynomial
time for arbitrary nonnegative capacities on $E_X$, then
\eqref{eq:cut-packing} has a polynomial-time computable optimum whose
positive support is laminar.
\end{lemma}

A proof, including polynomial-bit complexity, is given in
Appendix~\ref{app:laminarization}.  For $\mathcal A(X)$ the required
oracle is a global minimum cut; the restricted family is treated in
\cref{sec:approximation}.

Let $y$ be a feasible packing with laminar support
$\mathcal L:=\operatorname{supp}(y)=\{S\in\mathcal F:y_S>0\}$.  Its \emph{weighted laminar tree} $\mathcal T_y$ is the inclusion tree
of $\mathcal L$ augmented by $X$ and the missing singletons: each proper
set is joined to its smallest strict superset.  The edge for
$S\in\mathcal L$ has weight $y_S$; added edges have weight zero.
The singleton $\{v\}$ is the leaf labelled $v$, and total weight is
$\sum_{S\in\mathcal F}y_S$.  Whenever a
binary tree is required, we refine every node of outdegree greater than
two by zero-weight binary gadgets and retain the notation $\mathcal T_y$.
Every internal node of the resulting rooted binary tree has exactly two
children; no new labelled leaves are added.  This refinement changes
neither the total weight nor any weighted parity-join value.

For any tree $\mathcal T$ with labelled leaf set $X$, choose an arbitrary
root and, for each edge $e$, let $L_e\subseteq X$ be the labelled leaves
in the component of $\mathcal T-e$ not containing the root.  For every
even $T\subseteq X$, define
$J_{\mathcal T}(T):= \{e\in E(\mathcal T):|L_e\cap T|\equiv1\pmod2\}.$
This is the unique edge set whose odd-degree labelled leaves are exactly
$T$ and whose internal vertices have even degree; in particular, it is
independent of the chosen root.  Rooting $\mathcal T_y$ at its node
representing $X$, the positive edge corresponding to
$S\in\mathcal L$ is present precisely when $S$ is $T$-odd.  For a
weighted tree, write $w(F):=\sum_{e\in F}w(e)$.  Therefore,
\begin{equation}
  \score_y(T)=w(J_{\mathcal T_y}(T)),
\label{eq:tree-score}
\end{equation}

Fix a nonnegatively weighted tree $\mathcal T$ with labelled leaf set
$V$, an even reference set $A\subseteq V$, and a set $U\subseteq V$.
For $r\in\mathcal R_U:=\{0,2,\ldots,2\lfloor |U|/2\rfloor\}$,
define the exact-cardinality profile
$\mathcal G_{\mathcal T,A,U}(r):=
\max_{D\subseteq U,\ |D|=r}w(J_{\mathcal T}(A\triangle D))$.
The role of concavity is to convert a distribution over even sets of
varying cardinalities into an exact-cardinality solution.  Thus, an
expected tree-value guarantee can be preserved at the required cardinality.
\begin{lemma}[Exact-cardinality concavity]
\label{lem:concavity}
Let $\mathcal T$ be a nonnegatively weighted tree with labelled leaf set
$V$, let $A\subseteq V$ be even, and let $U\subseteq V$.  Then, for every
$r\in\mathcal R_U$ such that $r-2,r+2\in\mathcal R_U$,
\[
  2\mathcal G_{\mathcal T,A,U}(r)
  \ge
  \mathcal G_{\mathcal T,A,U}(r-2)
  +
  \mathcal G_{\mathcal T,A,U}(r+2).
\]
Equivalently, the piecewise-linear interpolation of
$\mathcal G_{\mathcal T,A,U}$ between consecutive points of $\mathcal R_U$ is
concave.
\end{lemma}

The proof is given in \cref{sec:structural-proofs}.

\begin{corollary}[Mean-cardinality rounding]
\label{cor:mean-rounding}
Extend $\mathcal G_{\mathcal T,A,U}$ linearly between consecutive points
of $\mathcal R_U$.  If $D\subseteq U$ is a random even set and
$\bar q:=\mathbb E|D|$, then
\[
  \mathcal G_{\mathcal T,A,U}(\bar q)
  \ge \mathbb E\bigl[w(J_{\mathcal T}(A\triangle D))\bigr].
\]
In particular, if $\bar q\in\mathcal R_U$, an exact-$\bar q$ set attains
at least the expected value.
\end{corollary}

\begin{proof}
By \cref{lem:concavity}, the interpolated profile is concave.  Hence
Jensen's inequality and the definition of the profile give
\[
 \mathcal G_{\mathcal T,A,U}(\bar q)
 \ge \mathbb E[\mathcal G_{\mathcal T,A,U}(|D|)]
 \ge \mathbb E[w(J_{\mathcal T}(A\triangle D))].
\]
\end{proof}

The following tight profile is the quantitative rounding statement used
by the centered all-cut algorithm in \cref{sec:approximation}.

\begin{theorem}[Tight density profile]
\label{thm:density}
Let $\mathcal T$ be a nonnegatively weighted rooted binary tree, every internal node having
exactly two children, with $m$ labelled leaves and total edge weight $W:=\sum_{e\in E(\mathcal T)}w(e)$.  For
every even $q$ with
$2\le q\le m$, an exact-$q$ leaf set $S$ can be found such that, for
$p=q/m$,
\begin{equation}
  w(J_{\mathcal T}(S))\ge
  \min\{p,2(1-p)\}W.
\label{eq:density-profile}
\end{equation}
Both branches of this bound are tight in general.
\end{theorem}

The proof is given in \cref{sec:structural-proofs}.

\begin{lemma}[Exact-cardinality tree dynamic program]
\label{lem:tree-dp}
Let $\mathcal T$ be a binary weighted tree with labelled leaf set $V$,
let $A\subseteq V$ be even, let $U\subseteq V$, and let
$q\in\mathcal R_U$.  A set
\[
 D^*\in\arg\max_{\substack{D\subseteq U\\|D|=q}}
      w(J_{\mathcal T}(A\triangle D))
\]
can be computed in polynomial time.
\end{lemma}

The recurrence, its initialization, and the proof appear in
Appendix~\ref{app:tree-dp-details}.

For nonempty $X\subseteq V$ and $1\le \ell\le |X|$, denote the
farthest-first traversal by $\FF(X,\ell)$~\cite{Gonzalez85}.  Starting
from an arbitrary point of $X$,
the procedure repeatedly adds a point whose distance from the set of
points chosen so far is maximum, until $\ell$ points have been selected.
Write $C:=\FF(X,\ell)=\{c_1,\ldots,c_\ell\}$.
The initial point and all ties are fixed by the input order; only
points not already selected are eligible.

\begin{lemma}[Farthest-first separation]
\label{lem:ff-separation}
If $C=\FF(X,\ell)$ and $R=\max_{v\in X}d(v,C)$, then
$d(c_i,c_j)\ge R$ for all distinct $c_i,c_j\in C$.
\end{lemma}

\begin{proof}
Assume $i<j$.  When $c_j$ is selected, its distance from
$\{c_1,\ldots,c_{j-1}\}$ is at least the final covering radius $R$.
In particular, $d(c_i,c_j)\ge R$.
\end{proof}

\section{Approximation Algorithms for Prescribed Cardinality}
\label{sec:approximation}

Throughout this section, $k$ is even and $0\le k\le n$.  All cut
families, induced metrics, parity scores, laminar trees, and
farthest-first sets use the notation fixed in \cref{sec:framework}.

The endpoint $k=0$, and also $k=n$ when feasible, has a unique solution.
Assume henceforth that $2\le k<n$.  The following algorithm has a center-set
size parameter $m\in\{k,\ldots,n\}$; its two choices used below are
$m=n$ and $m=3k/2$.

\par\medskip
\noindent\fbox{%
\begin{minipage}{\dimexpr\linewidth-2\fboxsep-2\fboxrule\relax}
\refstepcounter{algorithm}
\noindent\textbf{Algorithm~\thealgorithm: Centered all-cut algorithm}
\label{alg:centered}\par
\medskip
\begin{itemize}
\renewcommand{\labelitemi}{--}
\setlength{\topsep}{0pt}\setlength{\partopsep}{0pt}
\item Set $C:=\FF(V,m)$.

\item Compute an optimal laminar all-cut packing
$y^C$ on $(C,d_C)$.

\item Construct its weighted laminar tree $\mathcal T_C:=\mathcal T_{y^C}$.

\item Using the exact-cardinality tree dynamic program of
\cref{lem:tree-dp}, compute and return
\begin{equation}
 T_m\in
 \arg\max_{\substack{T\subseteq C\\|T|=k}}
 w\bigl(J_{\mathcal T_C}(T)\bigr).
\label{eq:centered-algorithm}
\end{equation}
\end{itemize}
\end{minipage}}
\par\medskip
By \cref{lem:laminarization,lem:tree-dp}, the algorithm is deterministic
and runs in polynomial time.

\begin{theorem}[Centered all-cut guarantee]
\label{thm:centered}
Put $\theta:=k/m$ and $h(\theta):=\min\{\theta,2(1-\theta)\}$.
Let $R:=\max_{v\in V}d(v,C)$.
The set returned by \eqref{eq:centered-algorithm} satisfies
\[
 \mu(T_m)\ge
 \max\left\{
   \frac{kR}{2},
   h(\theta)(\OPT_k-kR)
 \right\}.
\]
\end{theorem}

\begin{proof}
Let $O$ be an optimal $k$-set, map every point of $O$ to a nearest
center, and let $E\subseteq C$ contain the centers with odd preimage
multiplicity.  By \cref{lem:projection}, $E$ is even and
$\mu(O\triangle E)\le kR$.  The matching triangle inequality and the
all-cut packing certificate therefore give
\begin{equation}
 W_C:=W(C,\mathcal A(C))\ge\mu(E)\ge\OPT_k-kR.
 \label{eq:centered-projection}
\end{equation}

The tree $\mathcal T_C$ has $m$ labelled leaves and total edge weight
$W_C$.  By \cref{thm:density}, the dynamic program, \eqref{eq:tree-score},
and the packing certificate,
\[
 \mu(T_m)\ge h(\theta)W_C
             \ge h(\theta)(\OPT_k-kR).
\]
On the other hand, \cref{lem:ff-separation} makes the centers pairwise
$R$-separated.  Every perfect matching on the $k$ centers of $T_m$ has
$k/2$ edges, and hence $\mu(T_m)\ge kR/2$.
\end{proof}

\begin{corollary}[Density-dependent all-cut guarantee]
\label{thm:full}
For $0<k<n$, choosing $m=n$ gives value fraction
$\min\{p,2(1-p)\}$, where $p=k/n$, and its reciprocal as approximation
factor.  The feasible endpoints $k=0,n$ are exact.
\end{corollary}

\begin{proof}
For $m=n$ we have $C=V$, $R=0$, and $\theta=p$, so the claim follows
from \cref{thm:centered}.
\end{proof}

\begin{corollary}[Sparse factor-$7/2$ algorithm]
\label{thm:sparse}
If $2\le k\le2n/3$, choosing $m=3k/2$ gives a factor-$7/2$
approximation.
\end{corollary}

\begin{proof}
Here $\theta=2/3$.  Take the convex combination with weights $4/7$ and
$3/7$ of the two bounds in \cref{thm:centered}:
\[
 \mu(T_m)
 \ge \frac47 \times \frac{kR}{2}
      +\frac37 \times \frac23(\OPT_k-kR)
 =\frac27\OPT_k.
\]
\end{proof}

Let $A\subseteq V$ be even, let $U\subseteq V$, and let $s>0$ be even
with $2s\le |U|$.  Define $\Psi(A,U,s):=\max_{D\subseteq U,\ |D|=s}\mu(A\triangle D)$.
Because both $A$ and every feasible $D$ are even, $A\triangle D$ is
even and the objective is well defined.  Set
$C:=\FF(U,2s).$

Define the family of \emph{allowed cuts} by
\begin{equation}
  \mathcal F_{A,C}:=
  \bigl\{S\in\mathcal A(V):
    0<|S\cap C|<|C|
    \text{ or }|S\cap A|\equiv1\pmod2
  \bigr\}.
    \label{eq:allowed}
\end{equation}
Thus an allowed cut either splits $C$ or is $A$-odd.
Evenness of $A$ makes allowedness invariant under complementation;
it is independent of the chosen representative shore.

\begin{lemma}[Allowed-cut family]
\label{lem:allowed-family}
The family $\mathcal F_{A,C}$ is uncrossable.  Moreover, a
minimum-capacity member of $\mathcal F_{A,C}$ can be found in polynomial
time for arbitrary nonnegative edge capacities.
\end{lemma}

\begin{proof}
Call a shore \emph{forbidden} if it does not split $C$ and is $A$-even.
Forbiddenness is complement-invariant and preserved under disjoint
union: two disjoint forbidden shores cannot both contain $C$, and their
union neither splits $C$ nor has odd intersection with $A$.

Let $S$ and $T$ be crossing allowed shores, and write
\[
P=S\cap T,\qquad Q=S\setminus T,\qquad
R=T\setminus S,\qquad Z=V\setminus(S\cup T).
\]
If neither $(P,S\cup T)$ nor $(Q,R)$ consists of two allowed cuts,
complement invariance makes one of $P,Z$ and one of $Q,R$ forbidden.
They are disjoint, and their union is one of
\[
P\cup Q=S,\quad P\cup R=T,\quad
Z\cup Q=V\setminus T,\quad Z\cup R=V\setminus S.
\]
Closure of forbidden shores under disjoint union then contradicts the
allowedness of $S$ and $T$.  Thus $\mathcal F_{A,C}$ is uncrossable.

A cut splits $C$ iff it separates some $a,b\in C$.  Hence a minimum
such cut is the cheapest minimum $a$--$b$ cut over pairs in $C$.
Comparing it with a minimum $A$-odd cut from Padberg--Rao~\cite{PR82}
(or $+\infty$ if none exists) yields a minimum allowed cut.
\end{proof}

\refstepcounter{algorithm}
\par\medskip
\noindent\fbox{%
\begin{minipage}{\dimexpr\linewidth-2\fboxsep-2\fboxrule\relax}
\noindent\textbf{Algorithm~\thealgorithm: Normalized allowed-cut algorithm}\label{alg:quarter}\par
\smallskip
\noindent Input: even $A\subseteq V$, $U\subseteq V$, and even $s>0$ with $2s\le|U|$;\\
\phantom{Input: }set $C:=\FF(U,2s)$.
\medskip
\begin{itemize}
\renewcommand{\labelitemi}{--}
\item By \cref{lem:allowed-family,lem:laminarization}, compute an optimal
packing $y$ over $\mathcal F_{A,C}$ with laminar positive support.

\item Put $W:=\sum_{S\in\mathcal F_{A,C}}y_S=W(V,\mathcal F_{A,C})$.

\item Construct its weighted laminar tree $\mathcal T_y$.

\item Using \cref{lem:tree-dp}, compute {
\begin{equation}
    D_{\mathrm q}
    \in
    \arg\max_{\substack{D\subseteq U\\ |D|=s}}
    w\!\left(J_{\mathcal T_y}(A\triangle D)\right),
    \label{eq:quarter-dp}
\end{equation}}
\item Return $T_{\mathrm q}:=A\triangle D_{\mathrm q}$.
\end{itemize}
\end{minipage}%
}
\par\medskip
The returned set is even.  The algorithm is deterministic and runs in
polynomial time.

\begin{theorem}[Normalized quarter guarantee]
\label{thm:quarter-normalized}
The set returned by Algorithm~\ref{alg:quarter} satisfies
\[
    \mu(T_{\mathrm q})
    \ge \frac14\Psi(A,U,s).
\]
\end{theorem}

\begin{proof}
Choose $D$ uniformly from the even subsets of $C$; since $|C|=2s$,
$\mathbb E|D|=s$.  Every allowed cut is $(A\triangle D)$-odd with
probability at least $1/2$.  Indeed, if $S$ splits $C$, choose
$a\in S\cap C$ and $b\in C\setminus S$.  The involution $D\longmapsto D\triangle\{a,b\}$ preserves evenness and reverses the parity of
$|(A\triangle D)\cap S|$, since exactly one of $a,b$ lies in $S$;
thus exactly half of the choices are odd.
If $S$ does not split $C$, it is $A$-odd; since the even set $D\subseteq C$
lies entirely on one side, the cut remains $(A\triangle D)$-odd.
Thus $\mathbb E[\score_y(A\triangle D)]\ge W/2$.
By \eqref{eq:tree-score}, \cref{cor:mean-rounding}, the dynamic program,
and the packing certificate,
\[
 \mu(T_{\mathrm q})\ge\score_y(T_{\mathrm q})
 \ge\mathbb E[\score_y(A\triangle D)]\ge W/2.
\]

Put $R:=\max_{u\in U}d(u,C)$.  Every $C$-odd cut is allowed, and
\cref{lem:ff-separation} makes the points of $C$ pairwise at distance
at least $R$.  Hence $W\ge\mu(C)\ge sR$.

Let $D^\star$ attain $\Psi(A,U,s)$, map each of its points to a nearest
center in $C$, and let $E\subseteq C$ be the set of centers having odd
preimage multiplicity.  By \cref{lem:projection}, $\mu(D^\star\triangle E)\le sR$.
Since $E$ is even, every $(A\triangle E)$-odd cut is allowed: a cut not
splitting $C$ receives no parity contribution from $E$ and must therefore
be $A$-odd.  The packing certificate gives
$\mu(A\triangle E)\le W$.  Thus, by the matching triangle inequality,
\[
\Psi(A,U,s)
=\mu(A\triangle D^\star)
\le \mu(A\triangle E)+\mu(D^\star\triangle E)
\le W+sR
\le 2W.
\]
Combining the two bounds proves the claim.
\end{proof}

\begin{corollary}[Universal quarter guarantee]
\label{cor:quarter-global}
For every feasible even $k$, there is a deterministic polynomial-time
algorithm returning an exact-$k$ set $T$ such that
\[
    \mu(T)\ge\frac14\OPT_k.
\]
\end{corollary}

\begin{proof}
The case $k=0$ is trivial.  For $2\le k\le n/2$, use
$(A,U,s)=(\varnothing,V,k)$, giving $\Psi(A,U,s)=\OPT_k$.  If $k>n/2$,
take $U=V$ when $n$ is even and try every $U=V\setminus\{u\}$ when $n$ is
odd; set $A=U$ and $s=|U|-k$, running the normalized algorithm for $s>0$
and evaluating $U$ directly for $s=0$.  These choices satisfy $2s\le|U|$.
For an optimal $k$-set $O$, some such $U$ contains $O$; then
$D=U\setminus O$ has size $s$ and $A\triangle D=O$, while every feasible
$D'$ yields a $k$-set $A\triangle D'$.  Thus that branch has optimum
$\OPT_k$, and \cref{thm:quarter-normalized} proves the claim.
\end{proof}
\begin{theorem}[Overall guarantee]
\label{thm:overall}
For $0<k<n$, there is a deterministic polynomial-time $\rho(k/n)$-
approximation for exact-$k$ remote matching, with $\rho$ as in
\eqref{eq:intro-rho}.  The feasible endpoints $k=0,n$ are exact.
\end{theorem}
\begin{proof}
Solve the feasible endpoints directly.  Otherwise, run the centered
all-cut algorithm with $m=n$ and, when $p\le2/3$,
with $m=3k/2$; also run the universal factor-$4$ algorithm, and return
the best solution.  The value
guarantees $\min\{p,2(1-p)\}$, $2/7$ (when $p\le2/3$), and $1/4$
give exactly the four branches of \eqref{eq:intro-rho}.
\end{proof}

\section{Proofs of the Structural Results}
\label{sec:structural-proofs}

We prove concavity and the density bound from \cref{sec:framework}.
The routine dynamic-programming proof is in Appendix~\ref{app:tree-dp-details}.

\par
\begin{trivlist}
\item[\hskip\labelsep{\itshape Proof of \cref{lem:concavity}.}]
Write $g:=\mathcal G_{\mathcal T,A,U}$, and let $B,C\subseteq U$ attain
$g(r-2)$ and $g(r+2)$, respectively.  Put $\Delta:=B\triangle C$.
Membership of an edge in a tree join is determined by the parity of
the selected leaves on either side, so tree joins are linear over
symmetric difference:
\begin{equation}
 J_{\mathcal T}(A\triangle B)\triangle
 J_{\mathcal T}(A\triangle C)
 =
 J_{\mathcal T}(B\triangle C)
 =
 J_{\mathcal T}(\Delta).
\label{eq:concavity-linearity}
\end{equation}
For each nontrivial component $K$ of $J_{\mathcal T}(\Delta)$, put
$a_K:=|(C\setminus B)\cap V(K)|$ and
$b_K:=|(B\setminus C)\cap V(K)|$.
Since $\sum_K(a_K-b_K)=|C|-|B|=4$, some component satisfies
$a_K>b_K$.
The odd-degree vertices of $K$ are $\Delta\cap V(K)$, so by the
handshaking lemma $a_K+b_K$ is even.  Hence
$a_K-b_K$ is a positive even integer; consequently $a_K\ge2$, and
$K$ contains distinct $x,y\in C\setminus B$.

Let $P$ be the unique $x$--$y$ path in $\mathcal T$, and set
$B':=B\cup\{x,y\}$ and $C':=C\setminus\{x,y\}$.  Then
$|B'|=|C'|=r$.  Since $x,y$ lie in one component of $J_{\mathcal T}(\Delta)$,
their unique path satisfies $P\subseteq J_{\mathcal T}(\Delta)$.
Since $J_{\mathcal T}(\{x,y\})=P$, linearity shows that the new joins
are obtained by toggling $P$ in both original joins.  Every edge of $P$
belongs to exactly one original join, so their total weight is preserved.
Since $B'$ and $C'$ are feasible $r$-sets,
\[
 2g(r)
 \ge w(J_{\mathcal T}(A\triangle B'))
     +w(J_{\mathcal T}(A\triangle C'))
 =g(r-2)+g(r+2),
\]
as required.
\end{trivlist}

\par
\begin{trivlist}
\item[\hskip\labelsep{\itshape Proof of \cref{thm:density}.}]
Use the full binary rooting from the theorem statement and apply
\cref{lem:concavity} with $A=\varnothing$ and $U=V$.  Let $g$ denote the
resulting exact-cardinality profile and its linear interpolation.

Suppose $p\le2/3$ and set $a:=p/(2(1-p))\le1$.
Construct a random edge set as follows.  At the root select both child
edges with probability $p$, and otherwise neither.  Elsewhere, if the
parent edge is absent, select both children with probability $a$ and
otherwise neither; if it is present, select one child uniformly.
Every internal vertex has even degree: it has degree $0$ or $2$ when
its parent edge is absent, and degree $2$ when its parent edge is
present.  Hence the odd-degree vertices are leaves; let $S$ be their
set.  The handshaking lemma makes $S$ even, and uniqueness of tree
joins identifies the selected edge set with $J_{\mathcal T}(S)$.

Root edges have marginal $p$.  Inductively, if a parent edge has
marginal $p$, each child edge has marginal
\[
(1-p)a+p\cdot\frac12=p.
\]
Thus every edge has marginal $p$; in particular every leaf belongs to
$S$ with probability $p$, and
\[
\mathbb E|S|=pm=q,
\qquad
\mathbb E\bigl[w(J_{\mathcal T}(S))\bigr]=pW.
\]
By \cref{cor:mean-rounding}, $g(q)\ge pW$.

It remains to consider $p\ge2/3$.  Applying the same construction at
density $2/3$ and \cref{cor:mean-rounding} gives
\begin{equation}
 g(2m/3)\ge 2W/3,
 \label{eq:density-anchor}
\end{equation}
where $g$ is understood via its piecewise-linear interpolation if
$2m/3$ is not an admissible even cardinality.

If $m$ is even, $g(m)=w(J_{\mathcal T}(V))\ge0$.  Concavity and
\eqref{eq:density-anchor} place $g$ above the chord through
$(2m/3,2W/3)$ and $(m,0)$.  Consequently,
\[
 g(q)\ge2(1-q/m)W=2(1-p)W.
\]

If $m>3$ is odd, let $B_e$ be the even shore of edge $e$; it has
$|B_e|\ge2$.  The edge lies in $J_{\mathcal T}(V\setminus\{v\})$
exactly when $v\in B_e$.  Averaging gives
\[
 \frac1m\sum_{v\in V}w(J_{\mathcal T}(V\setminus\{v\}))
 =
 \frac1m\sum_{e\in E(\mathcal T)}|B_e|w(e)
 \ge\frac{2W}{m}.
\]
Thus $g(m-1)\ge2W/m$.  Interpolating with \eqref{eq:density-anchor}
gives the same lower bound $g(q)\ge2(1-p)W$.
The case $(m,q)=(3,2)$ belongs to the first range.

Tightness is witnessed by the star and cherry constructions described in
Appendix~\ref{app:density-tightness}; they attain respectively the two
branches $pW$ and $2(1-p)W$.
\end{trivlist}

\subsubsection*{Acknowledgements: Generative AI Use Declaration.}
OpenAI ChatGPT was used in developing and preparing this work to explore algorithmic and reduction ideas, draft and check proof arguments, produce diagnostic code, locate literature, and assist with exposition and LaTeX, including the code plotting the stated approximation bound. The human authors have verified the correctness of the content and retain responsibility for the correctness, originality, references, and final content of the work.

\subsubsection*{Disclosure of Interests.}
The authors have no competing interests to declare that are relevant to the content of this article.

\clearpage
\appendix
\section{A One-Shot Factor-$4$ Family}
\label{app:MN-factor4}

We finally record a simple family showing that the one-shot randomized
skeleton procedure can lose a factor arbitrarily close to $4$ under a
valid implementation of its arbitrary parity-repair step.

Fix $m\ge1$, put $K=4m$, and partition $n=3K=12m$ points into
$K$ triples $C_1,\ldots,C_K$.  Let
\[
 A=C_1\cup\cdots\cup C_{2m},
 \qquad
 B=C_{2m+1}\cup\cdots\cup C_{4m},
\]
and set $L=m^2$.  For distinct points define
\[
 d(x,y)=
 \begin{cases}
  1/2, & x,y\in C_i\text{ for some }i,\\
  1,   & x,y\text{ lie in different triples on the same side},\\
  L,   & x\in A,\ y\in B\text{ or vice versa}.
 \end{cases}
\]
This is a metric.  Farthest-first selects exactly one point from each
triple: after a triple has been represented, its remaining points are at
distance $1/2$ from the skeleton, whereas every unrepresented triple has
distance at least $1$.  Hence the skeleton $Y$ contains $2m$ points on
each side and
\[
 \mu(Y)=K/2.
\]
Moreover, each Voronoi cell is exactly its triple.

There is a feasible $K$-set containing $2m-1$ points of $A$ and
$2m+1$ points of $B$.  Every perfect matching of such a set contains a
cross edge, and conversely one cross edge always suffices.  Therefore
\[
 L\le \OPT_K\le L+K/2,
\]
so $\OPT_K=(1+o(1))L$.

Now independently sample every point of $Y$ with probability $1/2$, and
let $Z_0$ be the resulting set.  The two parities
\[
 |Z_0\cap A|\bmod2,\qquad |Z_0\cap B|\bmod2
\]
are independent uniform bits.  Consider the valid parity-repair rule
that, when $|Z_0|$ is odd, deletes a sampled point from the side whose
sampled cardinality is odd.  After this repair, both sides are odd
precisely when both original parities were odd, an event of probability
$1/4$.

Padding adds points in pairs from a single Voronoi cell and therefore
does not change either side parity.  Consequently, with probability
$1/4$ the padded solution has value between $L$ and $L+K/2$, while
otherwise its value is at most $K/2$; the algorithm also retains $Y$.
Thus
\[
 \mathbb E[\mathrm{ALG}]
 =\frac14L+O(K).
\]
Since $L=m^2$ and $K=4m$,
\[
 \frac{\OPT_K}{\mathbb E[\mathrm{ALG}]}\longrightarrow 4.
\]

This example concerns one execution and the above admissible choice in
the arbitrary-deletion step; it does not claim a factor-$4$ lower bound
for every parity-repair rule or for the amplified algorithm obtained by
independent repetitions.

\section{Supplementary Structural Details}
\label{app:structural-details}

\subsection{Tightness of the density profile}
\label{app:density-tightness}
For the left branch of \cref{thm:density}, take an equal-weight star with
unit leaf edges and refine it to a binary tree by zero-weight edges.  Then
$W=m$, and every exact-$q$ set has join weight $q=pW$.  For the right
branch, let $m\ge4$ be even, pair the leaves into cherries, give each
cherry-to-center edge unit weight and its two leaf edges zero weight, and
again use zero-weight binary refinements.  Here $W=m/2$; for $p\ge2/3$,
an exact-$q$ set makes at most $m-q$ cherry edges odd, attained by omitting
$m-q$ leaves from distinct cherries.  Hence the optimum join weight is
$m-q=2(1-p)W$.

\subsection{Exact-cardinality tree dynamic program}
\label{app:tree-dp-details}
Root $\mathcal T$ at its full binary root, and let $L_v$ be the labelled
leaves below node $v$.  If $|D\cap L_v|=j$, the parent edge of $v$ belongs
to the join exactly when
\[
 b_v(j):=(|A\cap L_v|+j)\bmod2=1.
\]
Thus the toggle count already determines terminal parity; no separate
parity state is required.  Let $F_v(j)$ be the maximum weight of selected
edges strictly below $v$, over $D\subseteq U\cap L_v$ with $|D|=j$.
Store only $0\le j\le q$.  At a leaf $v$, set $F_v(0)=0$ and, when
$v\in U$ and $q\ge1$, $F_v(1)=0$; all other states have value $-\infty$.
For children $u,z$ of $v$, with edge weights $\ell_u,\ell_z$, the recurrence is
\[
 F_v(j)=\max_{a+b=j}
 \{F_u(a)+F_z(b)+\ell_u b_u(a)+\ell_z b_z(b)\}.
\]
The maximization ranges over feasible child states.  Each feasible toggle
set splits into its two child sets, proving that its value is represented.
Conversely, two feasible child states combine to a feasible parent state.
The only newly charged edges are the two child edges, whose membership
is fixed by their child counts.  Induction proves the state meaning and
hence the recurrence.  The root has no parent edge, and its reference
and toggle cardinalities are even, so $F_{\rm root}(q)$ is exactly the
required optimum.  Predecessor pointers recover a maximizing set.

There are $O(n(q+1))$ states and $O(n(q+1)^2)$ arithmetic operations,
where $n=|V|$.  All values are sums of input edge weights; rational inputs
therefore give polynomial bit complexity.  A one-leaf tree is handled
directly.  This proves \cref{lem:tree-dp}.

\section{Polynomial-Time Laminarization}
\label{app:laminarization}
We prove \cref{lem:laminarization}, including the algorithmic qualification.
The empty family or a one-point ground set has value zero, so assume
$N:=|E_X|\ge1$ and $\mathcal F\ne\varnothing$.  The packing polytope is
bounded: every variable appears in a nonempty cut and is bounded by an
edge cost.  The assumed minimum-cut oracle separates its covering dual;
separation--optimization equivalence also recovers an optimal packing
with polynomial support~\cite{GLS88}.

For a shore $S$, put $h(S)=|S|(|X|-|S|)$ and
$H(y)=\sum_{S\in\mathcal F}h(S)y_S$.  Among maximum-value packings,
choose one minimizing $H$.  If positive shores $S,T$ cross, uncrossability
supplies either $S\cap T,S\cup T$ or $S\setminus T,T\setminus S$ in
$\mathcal F$.  Shift $\lambda=\min\{y_S,y_T\}>0$ from $S,T$ to this
pair.  The corresponding cut-incidence inequalities hold coordinatewise,
so no edge load increases, and total packing mass is unchanged.
Writing $i,a,b,o$ for the positive sizes of
$S\cap T,S\setminus T,T\setminus S,X\setminus(S\cup T)$, respectively,
the decrease in $H$ is $2ab\lambda$ for the first pair and $2io\lambda$
for the second.  Both are positive.  This contradicts its choice and
proves laminarity.  Root-oriented crossing shores have a nonempty outside
region, since it contains the designated root.

To implement the secondary objective in polynomial time, scale the
rational distances to integers and let $D\ge1$ bound them.  A packing
vertex has at most $N$ positive variables.  Its basis matrix has integer
entries zero or one, and absolute determinant at most $N!$.
Consequently two distinct vertex values of the unperturbed mass differ
by at least $1/(N!)^2$.  Moreover, every feasible packing satisfies
\[
 0\le H(y)=\sum_{e\in E_X}\sum_{S:e\in\delta_X(S)}y_S
 \le \sum_{e\in E_X}d_X(e)\le ND.
\]
Choose $\eta=1/(4(N!)^2ND)$ and maximize
$\sum_S(1-\eta h(S))y_S$.  The perturbation cannot prefer a smaller
vertex mass over a larger one; among maximum-mass solutions it minimizes
$H$.  Its binary encoding length is polynomial.

The perturbed covering constraints are
$x(\delta_X(S))\ge1-\eta h(S)$.  They are separated by finding a
minimum member of $\mathcal F$ in nonnegative capacities $x_e+\eta$:
a violation exists exactly when that minimum is below one.  The same
oracle therefore solves the perturbed program and recovers a
polynomial-support optimal packing.  The strict uncrossing argument
makes its support laminar.  Finally scale the weights back; its
\emph{unperturbed} mass is optimal for the original program.  This proves
the lemma without assuming that an arbitrary uncrossing sequence has
polynomial length.

\end{document}